\documentclass[11pt,a4paper]{article}

\usepackage[T1]{fontenc}
\usepackage{lmodern}
\usepackage[margin=28mm]{geometry}
\usepackage{amsmath,amssymb,amsthm}
\usepackage{enumerate}
\usepackage{hyperref}

\theoremstyle{definition}
\newtheorem{definition}{Definition}
\theoremstyle{plain}
\newtheorem{theorem}[definition]{Theorem}
\newtheorem{lemma}[definition]{Lemma}
\newtheorem{corollary}[definition]{Corollary}

\title{Approximation and parameterized algorithms for covering
disjointness-compliable set families}
\author{\href{https://orcid.org/0000-0002-6629-3243}{Zeev Nutov}
        \qquad Anael Vaknin\\[3pt]
        \normalsize The Open University of Israel\\
        \small\texttt{nutov@openu.ac.il}\qquad\texttt{anaeln118@gmail.com}}
\date{}

\hypersetup{
  hidelinks,
  pdfauthor={Zeev Nutov and Anael Vaknin},
  pdftitle={Approximation and parameterized algorithms for covering disjointness-compliable set families},
  pdfkeywords={disjointness-compliable families, network design, approximation algorithms, parameterized algorithms}
}

\newcommand{\sem}{\setminus}
\newcommand{\subs}{\subseteq}
\newcommand{\empt}{\emptyset}
\newcommand{\opt}{\mathsf{opt}}
\newcommand{\smt}{\mathsf{smt}}
\newcommand{\dc}{disjointness-compliable}
\newcommand{\al}{\alpha}
\newcommand{\FF}{\mathcal{F}}
\newcommand{\XX}{\mathcal{X}}
\newcommand{\sfec}{\textnormal{\textsc{Set-Family Edge-Cover}}}

\begin{document}
\maketitle

\begin{abstract}
Balancing charges in Generalized Point-to-Point Connection (GP2P) and
meeting group quotas at several roots are instances of the same problem:
covering a disjointness-compliable set family. An $\al$-approximation for
the restricted one-terminal problem in each residual instance gives a
greedy $\max\{2,\al\}H_\tau$-approximation for $\tau$ terminals. For
edge costs, Zhao and Xiao subsequently reduced the logarithmic loss to
a constant multiplier. For node costs we prove the harmonic bound with a
deterministic oracle. Logarithmic loss is then unavoidable even when the
one-terminal problem is a shortest path.

In $O^*(3^\tau)$ time we solve proper families exactly and obtain ratio
$\al+1$ when overlapping feasible components with disjoint terminal sets
have feasible union. This condition holds for the multiroot applications.
For Multiroot Quota Tree we obtain the one-root ratio in the same time.
For GP2P, rooted Quota Tree is already the case of one negative node.
Conversely, on general graphs we can approach any fixed deterministic
one-root approximation ratio arbitrarily closely in fixed-parameter time
in the number of components containing a negative node in an optimum.
This gives a $(2+\epsilon)$-approximation for every fixed $\epsilon>0$.
\end{abstract}

\section{Introduction} \label{s:intro}

In Generalized Point-to-Point Connection (GP2P), each node $v$ has an
integer charge $b(v)$, and every component must have non-negative total charge. If a set has negative charge, then in
every partition of the set into two parts, at least one part has negative
charge. In Multiroot Covering Steiner, the component of each root must
meet its group quotas. If a set containing a root misses one of its
quotas, the part containing that root still misses it. In either problem
a feasible forest must cross every deficient set. We use this partition
property to pass from an approximation for one terminal to an approximation
for many.

An edge {\bf covers} a set if it has exactly one end in the set.
An edge set that covers every member of $\FF$ is an {\bf $\FF$-cover}.

\begin{center}
\fbox{\begin{minipage}{0.98\textwidth} \noindent
\underline{\sc Set-Family Edge-Cover} \\
{\em Input:} An undirected graph $G=(V,E)$ with non-negative edge costs $\{c_e:e\in E\}$,
and a set family $\FF$ on $V$. \\
{\em Output:} A min-cost forest $J\subs E$ that covers $\FF$.
\end{minipage}} \end{center}

We consider feasible instances with $\empt\notin\FF$. Numerical input
is encoded in binary unless stated otherwise.

\begin{definition}[Goemans \& Williamson \cite{GW}] \label{d:dc}
A set family $\FF$ is {\bf \dc} if
$A'\subs A\in\FF$ implies $A'\in\FF$ or $A\sem A'\in\FF$.
If $\FF$ is also symmetric, namely, $A\in\FF$ implies $V\sem A\in\FF$,
then $\FF$ is {\bf proper}.
\end{definition}

Goemans and Williamson \cite{GW} gave a primal-dual $2$-approximation
for proper families, later extended to uncrossable families \cite{WGMV}. For a general {\dc} family, even a single
inclusion-minimal member can give rise to problems as different as rooted
$k$-MST and Group Steiner. We therefore express the approximation ratio
in terms of the ratio for the one-terminal problem.

\medskip\noindent
{\bf The transfer theorem.}
Every inclusion-minimal member of a {\dc} family $\FF$ is a singleton (Lemma~\ref{l:disj}).
Its node is a {\bf terminal}. Write $T=\{t:\{t\}\in\FF\}$ and $\tau=|T|$.
For $t\in T$, the {\bf halo family} $\FF(t)$ consists of the sets in
$\FF$ whose only terminal is $t$. A {\bf restricted cover} of $\FF(t)$
covers those sets using no edge incident to another terminal.
After we contract the components of a partial solution, the uncovered
sets form the residual family on the contracted graph. We use two assumptions.

\noindent
{\bf Assumption 1.} Membership in $\FF$ can be tested in polynomial time.\\
{\bf Assumption 2.} In every residual instance, including the original,
a polynomial-time oracle finds an $\al$-approximate restricted cover for
any given terminal, or determines that none exists.

\begin{theorem} \label{t:main}
Under Assumptions 1 and 2, {\sfec} for a {\dc} family with $\tau\ge1$
terminals admits approximation ratio $\max\{2,\al\}H_\tau$, where
$H_\tau=\sum_{i=1}^\tau 1/i$. The bound holds in expectation if the
oracle returns a feasible cover whenever one exists and has expected ratio $\al$.
\end{theorem}

With edge costs the algorithm buys a cheapest terminal pair path or
restricted one-terminal cover. Its logarithmic loss is tight even on a
caterpillar. Zhao and Xiao \cite{ZX2026} subsequently obtained a constant
multiplier using a related local tree problem on classes of vertices. With node costs,
a spider augmentation gives the same $\max\{2,\al\}H_\tau$ bound for
a deterministic oracle. Here an $o(\log\tau)$ ratio is NP-hard even for
GP2P with one supplier and an exact local oracle (Theorem~\ref{t:node-costs}).

\medskip\noindent
{\bf Applications.}
In GP2P the terminals are the negative-charge nodes. Contraction adds
the charges of the nodes in each component. After deleting all other
negative nodes, covering the halo of $t$ amounts to finding a cheapest
tree rooted at $t$ that collects supply at least $-b(t)$. This is rooted
Quota Tree. Its $(2+\epsilon)$-approximation \cite{FSS} gives a direct
deterministic bound for GP2P, whose logarithmic approximability was
already known \cite{HKKN,FSS}.

\begin{corollary} \label{t:P2P}
For every fixed $\epsilon>0$, GP2P admits a deterministic polynomial-time
$(2+\epsilon)H_\tau$-approximation, where $\tau$ is the number of
negative-charge nodes. For positive total charge $B=b(V)$, one may also
use the bound $2+(2+\epsilon)H_{B+2}$. For $B=0$, ratio $2$ is known
\cite{HKKN}.
\end{corollary}

In Multiroot Covering Steiner each root $r\in R$ specifies how many
members of each of its groups its component must contain. The
one-terminal problem is ordinary rooted Covering Steiner \cite{GS}.

\begin{corollary} \label{t:MGS}
Multiroot Covering Steiner with $n$ nodes, $m$ distinct groups, and
maximum group size $M$ admits expected approximation ratio
$O(\log(2+n)\log(2+m)\log(2+M)\,H_{|R|})$.
\end{corollary}

Zhao and Xiao \cite{ZX2026} removed the harmonic factor for Multiroot
Covering Steiner. The same improvement applies when each root has a
monotone submodular requirement, for example a prescribed matroid rank
(Section~\ref{ss:submodular}).

\medskip\noindent
{\bf Parameterized algorithms.}
Call a connected node set feasible if it lies outside $\FF$.
For proper families, feasibility depends only on the terminals in the
set. For non-symmetric families we use the following condition on
connected node sets in the original graph.
\begin{equation*}\tag{M}\label{eq:merge}
\parbox{.84\linewidth}{The union of two overlapping feasible sets with
disjoint terminal sets is feasible.}
\end{equation*}
Here $O^*(\cdot)$ suppresses factors polynomial in the binary input length.

\begin{theorem} \label{t:fpt}
Let $\FF$ be {\dc} with $\tau$ terminals and suppose Assumption~1 holds.
In time $O^*(3^\tau)$ we can
\begin{enumerate}[(i)]
\item solve {\sfec} exactly if $\FF$ is proper, and
\item obtain an $(\al+1)$-approximation under Assumption~2
if $\FF$ satisfies \eqref{eq:merge}.
\end{enumerate}
Both bounds allow binary-encoded non-negative rational edge costs.
In (ii), the bound holds in expectation if the oracle always returns a
feasible cover when one exists and its approximation ratio holds in expectation.
\end{theorem}

Multiroot Quota Tree, Multiroot Group Steiner, and Multiroot Covering
Steiner satisfy \eqref{eq:merge} because growing a component preserves
every quota already met for its roots. In $O^*(3^{|R|})$ time,
Multiroot Covering Steiner thus admits expected ratio $1+\al$ for the
one-root ratio of Lemma~\ref{l:MGS}. For Multiroot Quota Tree we obtain
the one-root ratio $\al$ itself. If an optimum serves its roots in at
most $p$ components, ratio $(1+\eta)\al$ takes $2^{O(p/\eta)}N^{O(1)}$
time for $0<\eta\le1$ and binary input length $N$
(Corollary~\ref{c:block-root}).

GP2P can violate \eqref{eq:merge}. Two trees, each with one unit of
demand, may both use the same unit of supply. Each tree is feasible and
their union is deficient. Nevertheless, if some optimum has at most
$\tau$ components of strictly positive total charge, GP2P admits a
$3$-approximation in $O^*(3^\tau)$ time. This includes every instance
with $0<b(V)\le\tau$ (Corollary~\ref{c:GP2P-few}).
Exact optimization takes $O^*(3^{\tau+k})$ time, where $k$ counts
suppliers whose charge is smaller than the total demand. In particular,
charges in $\{-1,0,n\}$ admit an exact $O^*(3^\tau)$ algorithm
(Theorem~\ref{t:fpt-ptp}).

\smallskip\noindent
{\bf GP2P and rooted Quota Tree.}
Rooted Quota Tree is already GP2P with one negative node. Given root
$r$, profits $w$, and quota $q$, the root alone suffices if $w(r)\ge q$.
Otherwise put $b(r)=w(r)-q$ and $b(v)=w(v)$ elsewhere. The component of $r$ has non-negative charge
exactly when it collects profit at least $q$. All other nodes can be left
isolated. Thus the optimum costs agree, and any fixed-parameter
approximation for GP2P in $\tau$ gives the same polynomial-time ratio for
rooted Quota Tree. The next theorem gives the converse up to an arbitrarily
small loss.

\begin{theorem} \label{t:gp2p-fpt-approx}
Suppose rooted Quota Tree with binary non-negative integer profits admits
a deterministic polynomial-time $\al$-approximation, $\al\ge1$.
Let $p\ge1$ be a given upper bound on the number of components containing
a negative node in some optimum GP2P forest. For $0<\eta\le1$, GP2P
admits ratio $(1+\eta)\al$ on general graphs in time
$(2+\al p/\eta)^{O(p/\eta)}N^{O(1)}$, where $N$ is the binary input
length.
\end{theorem}

For $\tau>0$ we may take $p=\tau$. For fixed $\al$, setting $\eta=1/\al$
recovers ratio $\al+1$ in time $2^{O(p\log(2+p))}N^{O(1)}$.
With the rooted Quota Tree algorithm of Friggstad, Salavatipour, and
Sun \cite{FSS}, the ratio is $2+\epsilon$ in FPT time for every fixed
$\epsilon>0$. Their earlier GP2P algorithm attained this ratio in time
$n^{O(\tau/\epsilon)}$ on general graphs and in FPT time in fixed-dimensional
Euclidean metrics \cite[Theorem~4]{FSS}.
Appendix~\ref{ss:gp2p-fpt-single} gives a polynomial-time
$(4+\epsilon)$-approximation whenever some optimum has only $O(\log n)$
components containing a negative node.

\section{Properties of disjointness-compliable families} \label{s:properties}

Write $n=|V|$, $c(J)=\sum_{e\in J}c_e$, and $\opt$ for the optimum cost.
The uncovered sets are exactly the unions of whole components of
$(V,J)$. Thus any spanning forest of $(V,J)$ covers the same sets as $J$.

\begin{lemma} \label{l:disj}\label{l:components}
Let $\FF$ be {\dc}. Every member contains a terminal, and $J$ covers
$\FF$ exactly when no component of $(V,J)$ belongs to $\FF$.
\end{lemma}
\begin{proof}
Repeatedly partition a member into two nonempty parts and retain a part
in $\FF$. This ends at a singleton in $\FF$. Likewise, if a union of
components belongs to $\FF$, one of those components belongs to $\FF$.
Such a component is itself uncovered.
\end{proof}

Given a partial solution $J$, contract each component into one node.
A set of contracted nodes belongs to $\FF^J$ if their union belongs
to $\FF$. Lifting a partition or a complement to the original nodes
shows that contraction preserves disjointness-compliability and symmetry.
Thus residual membership takes one query, and the residual terminals are
exactly the components belonging to $\FF$. Each contains an original
terminal, so contraction cannot increase their number.

\begin{lemma} \label{l:A'}
If $A\in\FF$ and $B$ has no terminal, then $A\sem B\in\FF$.
If $\FF$ is proper, $A\in\FF$ exactly when $A\cap T\in\FF$.
\end{lemma}
\begin{proof}
In the partition of $A$ into $A\cap B$ and $A\sem B$, the first part
has no terminal. For a proper family, applying the same fact to the
complement shows that adding nonterminals also preserves membership.
\end{proof}

In particular, a feasible connected set containing exactly one terminal
$t$ gives a restricted cover: any halo set it fails to cross would
contain the whole connected set and only additional nonterminals,
contrary to Lemma~\ref{l:A'}. Buying a restricted cover
eliminates at least one terminal: a residual terminal containing $t$
alone among the previous terminals would be an uncovered halo set.
Joining $k\ge2$ terminals eliminates at least $k-1$, since they now
belong to a single component.

\section{The transfer theorem} \label{s:main}

\begin{proof}[Proof of Theorem~\ref{t:main}]
In each residual instance take the cheapest candidate among paths joining
two terminals and the restricted covers returned by the oracle. Contract
its edges and repeat. When no terminal remains, the purchased edges
cover the family by Lemma~\ref{l:components}. Return a spanning forest
of these edges.

Suppose $m$ terminals remain and let $z$ be the cost of the chosen
candidate. An original optimum completes the residual instance at cost
at most $\opt$. Consider a tree $H$ of such a completion containing
$k\ge1$ terminals. For $k=1$, $H$ is a restricted cover, so
$z\le\al c(H)$. For $k\ge2$, a doubled tour of $H$ visits the $k$
terminals cyclically. One of its $k$ segments has cost at most $2c(H)/k$,
so $kz\le2c(H)$. Summing over the trees gives
$mz\le\max\{2,\al\}\opt$. The successive values of $m$ decrease
strictly, giving the harmonic bound and at most $\tau(\tau+1)/2$ oracle
calls. When $m=1$ the local call costs at most $\al\opt$, so the
edge-cost ratio sharpens to $\al+\max\{2,\al\}(H_\tau-1)$.
With a randomized oracle, condition on the residual instance:
the expected minimum candidate cost is at most the expected cost of
each comparison candidate. The same sum bounds the expected cost at
that iteration, and each $m$ is visited at most once.
\end{proof}

\smallskip\noindent
{\bf A tight example.} Take a path $v_0,\ldots,v_\tau$ with edge costs
$2$, supply $\tau=2^h-1$ at $v_0$, and unit demand at each other node.
Attach a private leaf of supply $1$ to each $v_i$, $i>0$. Its edge costs
$2^{j+1}-1$ if $i$ is an odd multiple of $2^j$. Greedy buys the private
edges in order of increasing $j$, even with an exact oracle. After the
earlier levels have been bought, their nodes are balanced, and the
remaining demand nodes lie at multiples of $2^j$. Reaching another
demand or $v_0$ costs at least $2^{j+1}$, while other remaining demands
block access to their private leaves in a restricted call. Thus every
tie-breaking rule spends $2^{h-j-1}(2^{j+1}-1)=\Theta(\tau)$ per level.
The path meets every demand for $2\tau$, giving a greedy ratio of $\Omega(\log\tau)$.

\begin{theorem}[Node costs]\label{t:node-costs}
Let $\FF$ be {\dc} with $\tau\ge1$ terminals and let node costs be
nonnegative. Under Assumption~1, a deterministic $\al$-approximation
for the residual restricted covers gives ratio $\max\{2,\al\}H_\tau$. Unless $P=NP$, no polynomial-time algorithm
achieves ratio $o(\log\tau)$ even for GP2P with one supplier and an
exact restricted-cover oracle.
\end{theorem}
\begin{proof}
Pay $w(v)$ when vertex $v$ first becomes incident to a bought edge.
Each contracted component already paid for has cost zero.
Measure an augmentation by its cost divided by the number
of terminals it eliminates. In addition to the restricted covers, try
every center $x$ and every $k\ge2$ up to the number of reachable terminals:
join $x$ to its $k$ closest distinct terminals by shortest paths in node
costs, excluding $x$. Count $x$ at distance zero if it is a terminal.
The union costs at most $w(x)$ plus the sum of the $k$ distances and
eliminates at least $k-1$ terminals. Buy a candidate of minimum cost
per terminal eliminated.

Consider an optimal residual forest with $m$ terminals. A component
containing just one terminal gives a candidate costing at most $\al$
times its node cost. In every component containing at least two
terminals, the Klein--Ravi spider decomposition \cite{KR} groups them
into vertex-disjoint spiders, each joining $k\ge2$ marked terminals
by legs disjoint except at the center (see also \cite{N-spiders}).
One mark may be the center.
For a spider of cost $c$, the $k$ nearest-terminal distances sum
to at most the costs of its legs excluding the center. The candidate
therefore costs at most $c$, or $c/(k-1)\le2c/k$ per terminal eliminated. Averaging over all $m$
marks bounds the minimum density by $\rho\opt/m$, where
$\rho=\max\{2,\al\}$. A step eliminating $d$ terminals costs at most
$\rho\opt\,d/m\le\rho\opt(H_m-H_{m-d})$, with $H_0=0$.
The bound telescopes.

The logarithmic dependence is necessary for node costs, by the Set
Cover reduction for node-weighted Steiner Tree \cite{KR}. Given a Set
Cover instance on $\tau$ elements, make a zero-cost supplier of charge
$\tau$, one zero-cost node of charge $-1$ per element, and a unit-cost
node of charge zero per set. Join each set node to the supplier and
to its elements. A set cover connects all demands to the supplier at
the same cost. Conversely, the purchased set nodes of any feasible
forest cover every element. After any contraction there is at most one
positive node, whose supply equals the remaining demand. A restricted
call is therefore a shortest path to that node after deleting the
other terminals, or is infeasible. Set Cover hardness \cite{DS} rules
out an $o(\log\tau)$ polynomial-time ratio unless $P=NP$.
\end{proof}


\section{Applications of Theorem~\ref{t:main}} \label{s:appl}

We return to edge costs for the rest of the paper. A union of {\dc}
families is {\dc} by the partition property in each family.

\subsection{GP2P and Quota Tree} \label{ss:MST}

\begin{center}
\fbox{\begin{minipage}{0.98\textwidth} \noindent
\underline{\sc Generalized Point-To-Point Connection (GP2P)} \\
{\em Input:} A graph $G=(V,E)$ with non-negative edge costs and integer,
possibly negative, charges $\{b(v):v\in V\}$, with $b(V)\ge0$. \\
{\em Output:} A min-cost spanning subgraph $H=(V,J)$ such that every component
$C$ of $H$ has $b(C)\ge0$. \\
\underline{Set family:} $\FF=\{A:b(A)<0\}$, where $b(A)=\sum_{v\in A}b(v)$.
\end{minipage}} \end{center}

Spanning subgraphs include isolated nodes as components.
Even, Kortsarz, and Slany \cite{EKS} studied GP2P as an infinite-capacity
fixed-cost flow problem and gave a logarithmic approximation in the total
charge magnitude. Hajiaghayi, Khandekar, Kortsarz, and Nutov \cite{HKKN}
obtained an expected $O(\log\min\{n,b(V)+2\})$ ratio in pseudo-polynomial
time by dynamic programming on a probabilistic tree embedding \cite{FRT}.
Friggstad, Salavatipour, and Sun \cite{FSS} made this polynomial by cost
scaling. Zhao and Xiao \cite{ZX2026} subsequently obtained ratio
$8+\epsilon$ for both GP2P and Multiroot Quota Tree.

The local problem used by our algorithm is rooted Quota Tree.

\begin{center}
\fbox{\begin{minipage}{0.98\textwidth} \noindent
\underline{\sc Quota Tree} \\
{\em Input:} A graph with non-negative edge costs, a root $r$, non-negative
integer node profits $b(v)$, and an integer quota $k$. \\
{\em Output:} A min-cost tree containing $r$ whose node profit is at least $k$.
\end{minipage}} \end{center}

The restricted GP2P problem is rooted Quota Tree: contraction sums the
charges, and after deleting the other negative nodes the component of
$t$ must collect profit $-b(t)$ from the positive nodes. Membership in
$\FF$ is tested by summing charges. Friggstad, Salavatipour, and Sun
\cite[Theorem~10]{FSS} extend the $(2+\epsilon)$-approximation of
Arora and Karakostas \cite{AK} to binary profits. This proves the
first bound in Corollary~\ref{t:P2P}. The earlier factor $3$ follows
from \cite{Garg3,JMP}. For unweighted $k$-MST, Garg \cite{Garg}
obtained factor $2$ (see also \cite{BMWW}).

For the bound in terms of total charge, the preprocessing of \cite{HKKN}
limits the number of residual components. The degree guarantee of Lau
and Zhou \cite{LZ} reduces its cost from $4\opt$ to $2\opt$.

\begin{lemma} \label{l:GP2P}
Suppose $B=b(V)>0$.
Let $q^*$ be the minimum number of positive-charge components of an optimal GP2P solution,
including isolated positive nodes.
Then $1\le q^*\le B$.
Given an integer $q\ge q^*$, there is a polynomial-time preprocessing of cost at most
$2\opt$ that leaves at most $q+3$ nonzero-charge components,
of which at most $q+2$ have negative charge.
Consequently, any Quota Tree ratio $\al\ge2$ gives a GP2P ratio
$2+\al H_{q^*+2}$.
\end{lemma}
\begin{proof}
Since the non-negative integer component charges sum to $B$, we have
$q^*\le B$. Add a node $s$ of charge $-B$ with a zero-cost edge to
each positive node. On this balanced instance impose the degree bound
$q$ at $s$, apply the degree-bounded algorithm of Lau and Zhou
\cite{LZ}, and take an inclusion-minimal feasible forest $K$.

For $q\ge q^*$, an optimum original solution extends to a balanced solution of the same cost
and $s$-degree $q^*\le q$: join $s$ to one positive node in each positive-charge component.
The cut requirement of balanced GP2P is the indicator of the proper family
$\{A:b(A)\ne0\}$.
The $0$--$1$ skew-supermodular case of the Lau--Zhou algorithm therefore gives
$c(K)\le2\opt$ and $\deg_K(s)\le q+3$.

The cut relaxation has a polynomial-time separation oracle.
Indeed, a minimum cut of nonzero charge can be found among the fundamental cuts of
a Gomory--Hu tree \cite{GH} for the proposed non-negative edge capacities.
To see this, let $A$ be any nonzero-charge cut and remove all tree edges of weight
at most its capacity.
Every remaining tree component lies wholly in $A$ or wholly outside $A$.
If all removed fundamental cuts had charge zero, all these components, and hence $A$,
would have charge zero.
Thus a nonzero-charge fundamental cut has no larger capacity.
The same argument applies after contracting already chosen edges.

Every component of $K$ has charge zero.
Removing $s$ splits its component into at most $q+3$ trees whose charges sum to $B>0$.
Thus at most $q+3$ components have nonzero charge.
At least one has positive charge, so at most $q+2$ have negative charge.
All components of $K$ avoiding $s$ remain balanced.
An original optimum completes the residual instance for additional cost
at most $\opt$, since its feasible components can be merged.
For $q=q^*$, Theorem~\ref{t:main} completes the forest at cost at most
$\al H_{q^*+2}\opt$. If $q^*$ is unknown, try every integer from $1$
to the number of positive nodes, skip infeasible relaxations, and return
the cheapest feasible solution.
\end{proof}

Since $q^*\le B$, Lemma~\ref{l:GP2P} gives the second bound in
Corollary~\ref{t:P2P}.
For $B=0$, the proper-family $2$-approximation applies directly
\cite{GW,HKKN}.

The bound on all nonzero-charge components also permits exact residual optimization.

\begin{corollary} \label{c:GP2P-few}
If some optimal GP2P forest has at most $q$ positive-charge components
(counting isolated positive nodes), there is a deterministic
$3$-approximation in time $O^*(3^q)$ given $q$. In particular, the
running time is $O^*(3^\tau)$ if $q^*\le\tau$, and is polynomial
if $q^*\le a\log_2 n$ for fixed $a>0$.
\end{corollary}
\begin{proof}
For $B=0$ use the known factor $2$. For $B>0$, the preprocessing of
Lemma~\ref{l:GP2P} costs at most $2\opt$ and leaves at most $q+3$
nonzero-charge nodes after contraction. Theorem~\ref{t:fpt-ptp}
solves this residual instance exactly in $O^*(3^{q+3})$ time.
Its cost is at most $\opt$, as an original optimum completes the forest.
For the two promises, take $q=\tau$ or $q=\lfloor a\log_2 n\rfloor$,
respectively. The first applies whenever $0<B\le\tau$, since every
positive-charge component contributes at least one to $B$.
\end{proof}

\medskip

A related application uses a common non-negative profit function but different quotas.

\begin{center}
\fbox{\begin{minipage}{0.98\textwidth} \noindent
\underline{\sc Multiroot Quota Tree} \\
{\em Input:} A graph with non-negative edge costs and integer node profits $b(v)\ge0$,
roots $R$, and a quota $k_r>b(r)$ for each $r\in R$. \\
{\em Output:} A min-cost spanning subgraph such that the component $C_r$ of each root
$r$ has $b(C_r)\ge k_r$. \\
\underline{Set family:} $\FF=\bigcup_{r\in R}\{A:r\in A,\ b(A)<k_r\}$.
\end{minipage}} \end{center}

In Multiroot Quota Tree contraction adds profits and takes the maximum
of the quotas at the roots in the component. A component whose profit
falls short of this maximum is a residual terminal. After deleting the
other residual terminals, the problem is again rooted Quota Tree.
Theorem~\ref{t:main} gives ratio $(2+\epsilon)H_{|R|}$.


\subsection{Problems related to Group Steiner} \label{ss:groups}

\begin{center}
\fbox{\begin{minipage}{0.98\textwidth} \noindent
\underline{\sc Multiroot Covering Steiner} \\
{\em Input:} A graph $G=(V,E)$ with non-negative edge costs, roots $R\subs V$,
and, for each $r\in R$, groups $\XX^r\subs 2^V$ with integer requirements
$1\le k_X^r\le |X|$. \\
{\em Output:} A min-cost spanning subgraph $H$ such that, for every $r\in R$,
its component $C_r$ satisfies $|C_r\cap X|\ge k_X^r$ for every $X\in\XX^r$. \\
\underline{Set family:} $\FF=\bigcup_{r\in R}\{A:r\in A,
\ |A\cap X|<k_X^r\mbox{ for some }X\in\XX^r\}$.
\end{minipage}} \end{center}

\begin{lemma} \label{l:MGS}
The Multiroot Covering Steiner family is {\dc} and satisfies the two
oracle assumptions with expected one-root ratio
$\al=O(\log(2+n)\log(2+m)\log(2+M))$.
\end{lemma}
\begin{proof}
The part containing a deficient root still violates its quota.
For a residual terminal $C$, work in the original graph, give bought
edges cost zero, and delete the other residual terminal components.
Choose a root in $C$ and collect all requirements of roots in $C$,
keeping the largest quota for each repeated group. A tree for this
one-root instance contracts to a restricted cover. Conversely, a
restricted cover's component at $C$ expands to such a tree using
zero-cost edges. Working on original nodes preserves the number of
group members in each component. Feasibility is tested in the root's
connected component. Apply Gupta and Srinivasan \cite{GS} with the
tree embedding of \cite{FRT}.
\end{proof}

There are at most $|R|$ terminals, proving Corollary~\ref{t:MGS}.
Unit quotas give Group Steiner \cite{GKR} and the multiroot case is
a special case of Set Connectivity \cite{CEGS}. Zhao and Xiao
\cite{ZX2026} obtained the same asymptotic order as the one-root bound
for Multiroot Covering Steiner in polynomial time.


\subsection{Monotone submodular component requirements} \label{ss:submodular}

We assume polynomial-time value oracles for the following functions.

\begin{center}
\fbox{\begin{minipage}{0.98\textwidth} \noindent
\underline{\sc Multiroot Submodular Quota} \\
{\em Input:} A graph with non-negative edge costs, roots $R$, and, for each $r\in R$,
a non-negative integer-valued monotone submodular function $f_r$ and an integer quota $q_r$. \\
{\em Output:} A min-cost spanning subgraph such that the component $C_r$ of each root
$r$ satisfies $f_r(C_r)\ge q_r$. \\
\underline{Set family:} $\FF=\bigcup_{r\in R}\{A:r\in A,\ f_r(A)<q_r\}$.
\end{minipage}} \end{center}

Roots with $f_r(\{r\})\ge q_r$ can be omitted.
Put $Q=\sum_{r\in R}q_r$.

\begin{theorem} \label{t:submodular}
Suppose rooted Submodular Tree Cover on at most $n$ vertices and target at most $Q+1$
admits approximation ratio $\al$.
Then Multiroot Submodular Quota admits ratio $\max\{\al,2\}H_{|R|}$.
The Zhao--Xiao framework \cite{ZX2026} gives ratio $O(\al)$ for the
same application. With a randomized oracle of expected ratio $\al$ that
returns feasible trees, the greedy bound holds in expectation and the
$O(\al)$ bound holds with high probability. For polynomially bounded
$Q$, one may take $\al=O(\log^{2+\eta}(2+n)\log(2+Q))$ for each
fixed $\eta>0$.
\end{theorem}
\begin{proof}
The part of any partition containing a witnessing root still violates its requirement,
by monotonicity. Thus $\FF$ is {\dc}, and membership is polynomial.

Fix a residual terminal $C$ and remove all other residual terminals.
For a collection $S$ of the remaining contracted nodes, let $U(S)$ be their union in $V$.
The requirements belonging to roots in $C$ are equivalent to
\[
g_C(S)=Q_C,\qquad
g_C(S)=\sum_{r\in C\cap R}\min\{q_r,f_r(U(S))\},\qquad
Q_C=\sum_{r\in C\cap R}q_r,
\]
where the selected tree must contain the contracted node $C$.
The disjoint bundles preserve submodularity under $U$, as do truncation
and summation. Roots in other retained components already meet their
quotas. Thus the restricted problem is rooted Submodular Tree Cover
with target $Q_C\le Q$. To use an unrooted oracle, add
$\mathbf1[C\in S]$ to the function and one to the target. Evaluate on
the root's connected component to check feasibility, and apply
Theorem~\ref{t:main}.

In a local region of \cite{ZX2026} use the same function on the
disjoint classes, with $C$ the unique deficient class and with their
class-edge costs. The tree contains $C$, and all other classes already
satisfy their own roots. Hence the target is exactly the local
feasibility condition. Theorem~1 of \cite{ZX2026} gives $O(\al)$.
For a randomized oracle, taking the cheapest of independent repetitions
makes the $2\al$ bound hold with high probability throughout the
polynomially many calls, as in their Covering Steiner application.
For polynomial targets, the Submodular Tree Cover bound follows
from \cite[Theorem~1.3]{HKKN-STC}; see also \cite{CZ}.
\end{proof}

For example, if $f_r$ is the rank function of a matroid on $V$, the
component of $r$ must contain an independent set of size $q_r$.
Each root may have its own matroid. On feasible instances $Q\le n|R|$,
and the rank of a union of contracted components is still a submodular
function of those components.

\smallskip\noindent
{\bf Several charge scenarios.}
Suppose every component must have non-negative charge under each of
$d$ charge vectors. The deficient family is the union of the $d$ scalar
families, hence is {\dc}. After contraction, keep one residual terminal
$C$ and delete the others. Each remaining node outside $C$ has
non-negative charge in every coordinate. For every coordinate with
deficit $q_i=-b_i(C)>0$, let $w_i(S)$ be the supply in $S$ outside $C$.
A tree containing $C$ meets all deficits exactly when
$\sum_i\min\{q_i,w_i(S)\}=\sum_i q_i$, a submodular quota.
For polynomially many scenarios and polynomially bounded integer charges,
the target is polynomial. The same sum of truncated profits applies
on any one-deficient-class region of \cite{ZX2026}. Their framework
therefore gives ratio $O(\log^{3+\eta}(2+n))$ with high probability
for each fixed $\eta>0$.


\section{Partitioning the terminals} \label{s:fpt}
\subsection{Proper families}

Call $S\subs T$ admissible if $S\notin\FF$. By Lemma~\ref{l:A'},
membership in a proper family depends only on the terminals present.
Disjoint unions of admissible sets are admissible. Hence the terminal
sets of the components of any cover partition $T$ into admissible blocks.
Conversely, connecting each block of such a partition covers $\FF$:
even when the trees meet, each resulting component contains a union of
whole blocks and is feasible by Lemma~\ref{l:components}.

Write $\smt(S)$ for the minimum cost of a Steiner tree spanning $S$ in
$G$, allowing other terminals as Steiner nodes. The equivalence above
shows that the optimum equals the minimum sum of these costs over
admissible partitions. For $U\subs T$ let $D(U)$ be that minimum on $U$.
Set $D(\empt)=0$, fix $t_U\in U$ for nonempty $U$, and compute
\begin{equation}\label{eq:partition-dp}
 D(U)=\min_{t_U\in S\subs U,\ S\notin\FF}
       \{\smt(S)+D(U\sem S)\}.
\end{equation}
The block containing $t_U$ is unique in each partition. The
Dreyfus--Wagner table \cite{DW} computes all Steiner costs in
$O^*(3^\tau)$ time, and the partition recurrence examines at most
$3^\tau$ pairs $S\subs U\subs T$. Backtracking gives a minimum
cover, proving Theorem~\ref{t:fpt}(i).

For Steiner Forest, the admissible terminal sets are precisely the unions
of the prescribed demand groups, so the recurrence partitions those groups.

\subsection{Merging feasible components} \label{ss:merge}

We prove Theorem~\ref{t:fpt}(ii). Two consequences of
disjointness-compliability will be used separately. A set without terminals
is outside $\FF$, and disjoint sets outside $\FF$ have union outside $\FF$.
In particular, adjoining nonterminals to a feasible set preserves feasibility.

\smallskip\noindent{\bf One block.}
For each nonempty $S\subs T$, find a cheapest tree $K_S$ containing
exactly the terminals in $S$. If no such tree exists, put $w(S)=\infty$.
Otherwise contract $K_S$ to $s$. If $V(K_S)\notin\FF$, take $C_S=K_S$.
If $V(K_S)\in\FF$, use Assumption~2 at $s$, keeping the oracle component
containing $s$. Expand $K_S$ and take a spanning tree $C_S$ of the
expanded union. Put $w(S)=c(C_S)$,
or $\infty$ if the oracle reports infeasibility.
The other residual terminals are precisely $T\sem S$, so the restricted
oracle avoids them. Its component at $s$ is feasible: otherwise its node
set would itself be an uncovered member of the halo of $s$.
Thus $C_S$ is feasible and contains exactly the terminals $S$.

Let $H$ be any feasible tree with terminal set $S$.
We have $c(K_S)\le c(H)$, and $V(H)\cup V(K_S)$ is feasible since
$K_S$ adds only nonterminals to $H$. After contraction, $H$ is a
restricted completion of cost at most $c(H)$: any halo set it fails to
cross would contain the feasible node set of the contracted witness and only additional
nonterminals. Consequently
\[
 w(S)\le c(K_S)+\al c(H)\le(1+\al)c(H).
\]
For a singleton $S$ the first term is zero.

\smallskip\noindent{\bf Combining the blocks.}
Use recurrence~\eqref{eq:partition-dp} with $w(S)$ in place of $\smt(S)$
and allow every nonempty block, including singletons.
In the union of the chosen $C_S$, each connected component is obtained by
successively joining intersecting trees. Their assigned terminal sets are
disjoint throughout, so \eqref{eq:merge} keeps each union feasible.
Unused vertices are nonterminals. Lemma~\ref{l:components} gives a cover.
For the cost bound, compare with the fixed partition induced by an optimum
forest. Summing the one-block inequalities gives $(1+\al)\opt$.
For randomized oracles, the same comparison takes expectations over
the costs of the blocks of this fixed partition.

\smallskip\noindent{\bf Computing all the trees in one table.}
Here terminals outside $S$ must be excluded even as Steiner nodes.
Give each edge $uv$ both orientations and give the arc $uv$ cost
$c(uv)+M\mathbf 1[v\in T]$, with $M=1+\sum_{e\in E}c_e$.
In an out-tree rooted at $t\in S$, every terminal other than $t$ incurs
the penalty exactly once. Thus a cheapest out-tree spanning $S$ first
minimizes the number of terminals used and then its original edge cost.
It has exactly terminal set $S$ if and only if its cost is less than
$|S|M$. Subtracting $(|S|-1)M$ then gives $c(K_S)$.

For completeness, let $F(v,S)$ be the least modified cost of an out-tree
rooted at $v$ reaching $S$, and let $d$ denote directed shortest-path costs.
Start with $F(v,\{t\})=d(v,t)$. For $|S|\ge2$ the Dreyfus--Wagner
recursion is
\[
 F(v,S)=\min_x\bigg[d(v,x)+
 \min_{\empt\ne A\subsetneq S}\bigl(F(x,A)+F(x,S\sem A)\bigr)\bigg].
\]
Split a pruned optimum at its first branch or required terminal,
splitting off the singleton in the latter case. Conversely, the two
subtrees and the initial path contain a feasible out-tree. The table
takes $O^*(3^\tau)$ time and $F(t,S)$ gives the required cost for
each $S$. There are at most $2^\tau-1$ restricted-cover calls.

\subsection{Multiroot applications} \label{ss:block-oracles}

In each multiroot problem of Sections~\ref{ss:MST}--\ref{ss:submodular},
a root's requirement is monotone in its component. Every root in the union
of two feasible sets therefore keeps its requirement satisfied.
Theorem~\ref{t:fpt}(ii) applies with the corresponding residual oracle.
For Multiroot Group and Covering Steiner this gives expected ratio
$1+\al$ with $\al$ as in Lemma~\ref{l:MGS}, in $O^*(3^{|R|})$ time.

We can also ask directly for one tree containing a nonempty $S\subs R$
and meeting the requirements of all roots in $S$. It may traverse other
roots. If this problem admits ratio $\beta$, the partition recurrence
gives ratio $\beta$ in $O^*(3^{|R|})$ time. Every root stays served in
the union, and each component of an optimum is a comparison tree for
its block. This comparison also holds in expectation.

\begin{corollary}\label{c:block-root}
Multiroot Quota Tree admits the rooted Quota Tree ratio $\al$ in
$O^*(3^{|R|})$ time, in particular $2+\epsilon$ for fixed $\epsilon>0$.
Multiroot Submodular Quota admits in the same time the rooted
Submodular Tree Cover ratio for targets at most $Q+|R|$, with the
same expected or deterministic qualification as the one-root oracle.
For polynomially bounded $Q$ this gives expected ratio
$O(\log^{2+\eta}(2+n)\log(2+Q))$ for each fixed $\eta>0$,
since there are at most $Q$ roots with positive quotas.

For Multiroot Group and Covering Steiner take the corresponding
one-root ratio on at most $m+|R|$ groups of maximum size $M$.

In all four problems, let $p\ge1$ bound the number of root-containing
components of some optimum. Given $p$ and $0<\delta\le1$, a factor
$1+\delta$ on the respective one-root ratio suffices in
$2^{O(p/\delta)}N^{O(1)}$ time, with the same oracle qualification.
In particular, for fixed $a,\epsilon>0$, Multiroot Quota Tree admits a
polynomial-time $(2+\epsilon)$-approximation if $p\le a\log n$.
\end{corollary}
\begin{proof}
For Quota Tree fix $s\in S$, put $q=\max_{r\in S}k_r$, and let
$B=1+\sum_v b(v)$. Add profit $B$ to each root in $S\sem\{s\}$ and ask
for a tree rooted at $s$ of profit at least $q+(|S|-1)B$.
Omitting even one such root loses more than all original profit combined.
A tree containing $S$ meets the new quota exactly when its original
profit is at least $q$. This single call has polynomial binary input
length. Equivalently, one can use weighted $k$-MST with required nodes
directly \cite[Theorem~10]{FSS}.

For a submodular instance use root $s\in S$ and the function
\[
 h_S(A)=|A\cap S|+\sum_{r\in S}\min\{q_r,f_r(A)\}.
\]
It is monotone and submodular. Its upper bound $|S|+\sum_{r\in S}q_r$
is attained precisely when $A$ contains $S$ and satisfies every
requirement in the block. For Group and Covering Steiner, collect all
group requirements of the roots in $S$, keeping the largest quota on
each group. Add a singleton group for each required root and use any
one as the root of the tree. These instances have the same feasible
trees as the block problem.

For the last bound start with singleton root groups. Repeatedly merge
the two groups with a closest pair of roots until at most
$q=\lceil p+2p/\delta\rceil$ groups remain, and run the partition
recurrence on these groups. For positive optimum set $d=\delta\opt/p$.
In the graph joining roots at distance at most $d$, representatives of
different components are more than $d$ apart. A doubled tour of an
optimum tree of cost $c$ visits at most $1+2c/d$ such representatives.
Thus there are at most $p+2\opt/d\le q$ components, and every merge
made by the algorithm uses a pair at distance at most $d$. For
$\opt=0$, use $d=0$, since the roots of each optimum tree are at distance zero.
Starting from the comparison optimum, join a merged pair by a shortest
path only if its endpoints are still disconnected. At most $p-1$ paths
are needed. The resulting components contain whole groups, still satisfy
all their roots, and cost at most $\opt+(p-1)d\le(1+\delta)\opt$.
Their partition gives the claimed comparison for the block oracle.
There are at most $2^q$ calls and $3^q$ partition transitions.
\end{proof}

\subsection{Exact GP2P} \label{s:appl-fpt}

Put $T=\{v:b(v)<0\}$ and let $D=-b(T)$ be the total demand.

\begin{theorem} \label{t:fpt-ptp}
Let $k$ count the nodes with $0<b(v)<D$. GP2P can be solved exactly
in time $O^*(3^{\tau+k})$, hence in $O^*(3^\kappa)$ for $\kappa$
nonzero-charge nodes. If every positive charge is at least $D$, this is
$O^*(3^\tau)$, in particular for charges in $\{-1,0,n\}$.
\end{theorem}

\begin{proof}
For $T=\empt$ return the empty edge set. Otherwise let
$P=\{v:b(v)\ge D\}$ and $Q=T\cup\{v:0<b(v)<D\}$.
A component meeting $P$ can cover all the demand in the instance.
For nonempty $S\subs Q$ use the following block cost in the
minimum-sum partition recurrence, allowing singleton blocks:
\[
 w(S)=\begin{cases}
 \smt(S), & b(S)\ge0,\\
 \min_{v\in P}\smt(S\cup\{v\}), & b(S)<0.
 \end{cases}
\]
All trees are computed in the full graph. An empty minimum is infinite.
In an optimum, the intersection of each component with $Q$, if nonempty,
is a block. When the block has negative charge its component contains
a node in $P$, so in either case $w(S)$ is at most the component cost.

Conversely, connect each block of a minimizing partition by its chosen tree.
Every component meeting $P$ is feasible. In any other component, the nodes
of $Q$ are a disjoint union of whole assigned blocks, even if their trees
intersect. Each such block has nonnegative charge, since a tree for a
negative-charge block must meet $P$. All other nodes in this component
have charge zero. This proves feasibility and hence optimality.
The ordinary Dreyfus--Wagner table on $Q$ gives the cost of spanning
$S\cup\{v\}$ for every $S\subs Q$ and $v\in V$ in $O^*(3^{|Q|})$ time.
The partition recurrence has at most $3^{|Q|}$ transitions.
\end{proof}

In a balanced instance every feasible component has charge zero, so
one can also use the proper family $\{A:b(A)\ne0\}$. Its terminals are
all $\kappa$ nonzero-charge nodes. Counting only negative nodes cannot
give an exact FPT algorithm unless $P=NP$: give one Steiner Tree terminal
charge $-k$, the other $k$ terminals charge $+1$, and all other nodes
charge zero. Feasibility then requires connecting all terminals.

\section{Parameterized approximation for GP2P} \label{s:gp2p-fpt}

We prove Theorem~\ref{t:gp2p-fpt-approx}. Write $T=\{v:b(v)<0\}$,
and let $p\ge1$ be a supplied upper bound on the number of components
meeting $T$ in some optimum. We may take $p=\tau$ when $T\ne\empt$.
For each group of negative nodes, the search tries accepting the tree
returned by the quota oracle. It also tries reserving a neighborhood
of a vertex of this tree for a later group. Separation bounds how many
reservations are needed.

\subsection{Quota calls and preprocessing} \label{ss:gp2p-ownership}

For $\empt\ne A\subs T$, $A\subs W\subs V$, and
$P\subs W\cap\{v:b(v)>0\}$, consider the following one-tree problem:
find a min-cost tree $K$ in $G[W]$ that contains $A$ and satisfies
\begin{equation} \label{eq:owned-quota}
b(V(K)\cap P)\ge-b(A).
\end{equation}
Only the supplies in $P$ are credited against the demand of $A$.
Other nodes, including negative nodes outside $A$, may be used as Steiner nodes.
A call is declared infeasible if a required node is outside $W$.

\begin{lemma} \label{l:required-quota}
An $\al$-approximation for rooted Quota Tree gives an $\al$-approximation
for \eqref{eq:owned-quota} by one call, with polynomial binary input length.
Infeasibility can be tested in polynomial time.
\end{lemma}
\begin{proof}
Choose a root $a\in A$ and put $M=1+\sum_{v:b(v)>0}b(v)$.
Give profit $M$ to every node of $A\sem\{a\}$, profit $b(v)$ to
every $v\in P$, and profit zero elsewhere.
Use quota $(|A|-1)M-b(A)$.
Omitting one of the required nodes loses more than the sum of all ordinary
profits, so every feasible tree contains $A$.
After these nodes are included, the remaining quota is exactly $-b(A)$.
The feasible trees and their costs are therefore unchanged.
All new numbers have polynomial bit length.
The instance is feasible exactly when the component of $a$ in $G[W]$
contains $A$ and enough profit from $P$.
\end{proof}

Contract each zero-cost component, sum its charges, and retain its
edges for expansion.
This preserves the optimum and the upper bound $p$, and cannot increase $\tau$.
If no negative node remains, the zero-cost edges give an optimum solution.
Otherwise $\opt>0$.
Enumerate scales $L$ by doubling from the least positive edge cost through
the first scale at least $c(E)$.
One scale satisfies $\opt\le L<2\opt$.
All distances below are in the original graph of this contracted instance,
also during calls on induced subgraphs.
Write $B(x,r)=\{v:\operatorname{dist}(x,v)\le r\}$.

We use the following elementary packing bound.
A set of nodes of a tree $H$ at pairwise graph distance greater than $r>0$
has size at most $1+2c(H)/r$.
Indeed, a doubled traversal of length $2c(H)$ visits them all, and consecutive
first visits are more than $r$ apart along the traversal.
A maximal such set can be found greedily, and its radius-$r$ balls cover $V(H)$.

\subsection{An approximation arbitrarily close to the quota ratio}
\label{ss:gp2p-fpt-close}

Fix $0<\eta\le1$ and put $r=\eta L/(8p)$.
We use the separation argument of Friggstad, Salavatipour, and Sun
\cite[Lemma~11]{FSS}, and then find suitable regions by branching on
nodes of trees returned by the quota oracle.

\medskip
\noindent
{\bf A separated comparison forest.}
For the analysis, take an optimum with at most $p$ components meeting $T$.
While two such components are at distance at most $4r$, join them by a
shortest path.
This preserves feasibility: the new component contains whole feasible
components and possibly additional nonnegative nodes.
There are at most $p-1$ additions.
Take a spanning tree on the full node set of each resulting component meeting $T$.
The obtained trees $H_1,\ldots,H_m$, $m\le p$, satisfy
\begin{equation} \label{eq:gp2p-separated}
\operatorname{dist}(H_i,H_j)>4r\quad(i\ne j),\qquad
C:=\sum_j c(H_j)\le\opt+4r(p-1)\le(1+\eta)\opt\le2L.
\end{equation}

\smallskip\noindent
{\bf Terminal groups.}
Join terminals at distance at most $4r$ in an auxiliary graph.
By separation each auxiliary component lies in one $H_j$. Choosing
one terminal per auxiliary component gives points at pairwise distance
greater than $4r$. The packing bound on the $H_j$ gives at most
$m+C/(2r)\le p+L/r$ such points. Discard the scale if the number of auxiliary components exceeds this bound.
Enumerate partitions of the auxiliary components into at most $p$
nonempty terminal groups, ordered by least vertex. One partition gives
$S_j=V(H_j)\cap T$ for every $j$, after reindexing. There are at most
$p^{\lceil p+L/r\rceil}$ assignments.

\smallskip\noindent
{\bf Search.}
For each partition process the groups in order. A branch stores its
accepted trees and a list of balls, each reserved for a group. Initially
both are empty. For group $i$, delete all vertices of accepted trees,
all terminals outside $S_i$, and all balls reserved for other groups.
In the remaining graph use Lemma~\ref{l:required-quota} to find a
quota tree $K$ spanning $S_i$, with all remaining supplies eligible.
Stop this branch on infeasibility or if $c(K)>2\al L$.
Otherwise take the following branches.
\begin{itemize}
\item Accept $K$ and proceed to the next group.
\item If fewer than $D=\lceil p+2L/r\rceil$ balls are reserved, take a
maximal $Z\subs V(K)$ at pairwise distance greater than $r$. For each
$z\in Z$ and $j>i$, reserve $B(z,3r)$ for group $j$ in a separate
branch and retry group $i$.
\end{itemize}
Distances always refer to the graph before these deletions. A branch that
accepts all groups yields disjoint trees, each meeting its own demand.
Return the cheapest such forest over all scales and partitions.

\smallskip\noindent
{\bf Correctness.}
Use the scale and partition corresponding to \eqref{eq:gp2p-separated}.
We follow a branch on which the accepted trees avoid every unprocessed
$H_j$, and each ball centered at $z$ reserved for $j$ has a witness
$y\in H_j$ at distance at most $r$ from $z$. By separation this ball
misses every other $H_h$. Thus $H_i$ is available and
$c(K)\le\al c(H_i)\le2\al L$. If $K$ misses all later trees, accept
it. Otherwise choose $y\in V(K)\cap V(H_j)$ for some $j>i$.
By maximality some $z\in Z$ lies within $r$ of $y$, so reserving
$B(z,3r)$ for $j$ preserves the invariant.

For a previous reservation for $j$, with center $z'$ and witness $y'$,
the current $K$ avoids $B(z',3r)$. Thus $y$ and $y'$ are more than $2r$
apart. All witnesses in each $H_j$ are pairwise more than $2r$ apart,
so there are at most $m+C/r\le D$ in total, counting the proposed
reservation. The cutoff therefore permits the branch to finish. Its
accepted trees cost at most $\al\sum_jc(H_j)\le\al(1+\eta)\opt$.

\smallskip\noindent
{\bf Running time.}
A branch has at most $p$ acceptances and $D$ reservations.
The packing bound gives $|Z|\le1+2c(K)/r\le1+4\al L/r$, so the
branching factor is at most $1+p(1+4\al L/r)$. Substituting
$L/r=8p/\eta$ and including the terminal assignments gives
$(2+\al p/\eta)^{O(p/\eta)}$ calls per scale. The remaining work
and the number of scales are polynomial in the binary input length.
This proves Theorem~\ref{t:gp2p-fpt-approx}.


\section{Concluding remarks} \label{s:cr}

Theorems~7 and~8 of the \href{https://arxiv.org/abs/2512.20180v1}{first
arXiv version} claimed $O^*(3^\tau)$ time more generally. The proof for
non-symmetric families requires the merge condition~\eqref{eq:merge},
which suffices for the multiroot applications. For balanced GP2P, the
former exact bound is false: even one negative node can encode Steiner
Tree (Section~\ref{s:appl-fpt}). For positive total charge, whether a
$4$-approximation is possible in $O^*(3^\tau)$ time remains open.
Can this running time be combined with ratio $\al+1$ for unrestricted
GP2P, or more generally under Assumptions~1 and~2 alone?

\medskip
\noindent\textbf{AI disclosure.}
The authors used no AI tools in the original submitted work. After
submission they used ChatGPT to explore corrections and proof strategies
and to assist in drafting and editing this revision.

\clearpage
\appendix
\section{A single-exponential approximation for GP2P}
\label{ss:gp2p-fpt-single}

\begin{theorem}\label{t:single-exp}
Under the rooted Quota Tree assumption of Theorem~\ref{t:gp2p-fpt-approx},
let $p\ge1$ bound the number of components containing a negative node in
some optimum GP2P forest. Given $p$ and $0<\epsilon\le1$, one can obtain
ratio $\al(2-1/p+\epsilon)$ in time $(2/\epsilon)^{O(p)}N^{O(1)}$.
\end{theorem}

Consequently, for fixed $a,\delta>0$ there is a polynomial-time
$(4+\delta)$-approximation if some optimum has at most $a\log_2 n$
components containing a negative node. Indeed, give the algorithm
$p=\max\{1,\lfloor a\log_2 n\rfloor\}$ and use the rooted Quota Tree
ratio $2+\epsilon$ \cite{FSS} with
$\epsilon=\min\{1,\delta/5\}$. Then $(2+\epsilon)^2\le4+\delta$.

We now prove Theorem~\ref{t:single-exp}. Use the preprocessing, scales
$L$, and packing bound of Section~\ref{ss:gp2p-ownership}. The algorithm assigns each supply
to at most one group, so trees for different groups may intersect.

At scale $L$ put $r=\epsilon L/(6p)$ and take a maximal set $X\subs T$
of terminals at pairwise distances greater than $r$. Every terminal is
within distance $r$ of $X$. Discard the scale if $|X|>p+2L/r$.
The scale with $\opt\le L<2\opt$ passes this test by the packing bound
on the at most $p$ optimum components meeting $T$.

For the analysis, join each optimum component meeting $T$ but missing $X$
to a node of $X$ by a path of cost at most $r$.
The resulting components are feasible since they contain whole optimum components.
Take a spanning tree on each one's full node set, obtaining
$H_1,\ldots,H_m$, $m\le p$.
Each contains a node of $X$, and
\begin{equation} \label{eq:gp2p-radius-budget}
C:=\sum_i c(H_i)\le\opt+pr.
\end{equation}
Choose distinct $x_i\in V(H_i)\cap X$ and put
$R_i=r\lceil c(H_i)/r\rceil$.
Then $H_i\subs B(x_i,R_i)$ and $\sum_iR_i\le C+pr$.

The algorithm enumerates at most $p$ distinct centers in $X$ and positive
integer radius multiples of $r$ whose sum in units of $r$ is at most
$K_0=\lfloor L/r\rfloor+2p$.
This includes the indicated choice.
Assign each node in the union of the balls to a containing ball of minimum
radius, breaking ties by a fixed ordering.
Discard a choice that leaves a negative node unassigned; unassigned positive
nodes receive no credit.

For a nonempty subset $I$ of the center indices, let $A_I$ and $P_I$ be the
negative and positive nodes assigned to centers in $I$.
In the full graph, use Lemma~\ref{l:required-quota} with required set $A_I$
and eligible supplies $P_I$; let $w(I)$ be the returned cost.
If $A_I=\empt$, put $w(I)=0$ and return an empty tree.
An infeasible call has value $\infty$.
Use the partition recurrence~\eqref{eq:partition-dp} on the centers
with costs $w(I)$, allowing every nonempty block. Every component of
the union of the chosen trees contains whole assigned demand blocks.
Their credited supplies are distinct and lie in that component, so their
quota inequalities sum to nonnegative charge. Return the cheapest union
over all choices, deleting cycles.

\medskip
\noindent
{\bf Approximation guarantee.}
Use the indicated centers and radii.
On indices $1,\ldots,m$, join $i$ to $j\ne i$ whenever a node
$v\in V(H_i)$ was assigned to $x_j$.
Since $B(x_i,R_i)$ also contains $v$, we have $R_j\le R_i$.
The edge has a realizing path from $v$ to $x_j$ of cost at most
$R_j=\min\{R_i,R_j\}$, joining $H_i$ to $H_j$.
Choose a spanning tree in each auxiliary component and root it at an index
of maximum radius in that component.
Charge each realizing path to its child index.
Their total cost $B$ satisfies
\begin{equation} \label{eq:gp2p-bridge-budget}
B\le\sum_i R_i-
 \sum_{\text{auxiliary components }I}\max_{i\in I}R_i
\le\left(1-\frac1p\right)\sum_iR_i
\le\left(1-\frac1p\right)(C+pr).
\end{equation}
For the middle inequality, at least one root radius is as large as
$\sum_iR_i/m$, and $m\le p$.

For each auxiliary component $I$, the union of its $H_i$ and realizing paths
is connected.
Every node of $H_i$ is assigned to a center in the same auxiliary component as $i$.
Also, a negative node assigned to a center in $I$ belongs to some $H_i$ with $i\in I$.
Consequently, $A_I$ is exactly the negative-node set of these comparison trees,
and all their positive nodes belong to $P_I$.
Their union therefore provides enough eligible supply for $A_I$.
Nodes traversed by the realizing paths require no additional credit.
This partition has total quota cost at most $\al(C+B)$.
By \eqref{eq:gp2p-radius-budget} and \eqref{eq:gp2p-bridge-budget},
\[
\begin{split}
C+B&\le\left(2-\frac1p\right)C+(p-1)r\\
&\le\left(2-\frac1p\right)\opt+(3p-2)r
\le\left(2-\frac1p+\epsilon\right)\opt,
\end{split}
\]
where the last inequality uses $r=\epsilon L/(6p)$ and $L<2\opt$.

\medskip
\noindent
{\bf Running time.}
For $m$ centers there are at most $\binom{|X|}{m}\binom{K_0}{m}$
choices: the second factor counts positive radii with total at most
$K_0$ units, in a fixed center order. Both upper indices are
$O(p/\epsilon)$. Using $\binom{a}{m}\le(ea/m)^m$ and
$(p/m)^m\le e^{p/e}$, the sum for $m\le p$ is $(2/\epsilon)^{O(p)}$.
Each choice takes at most $2^p$ quota calls and $3^p$ partition
transitions, proving Theorem~\ref{t:single-exp}.

\end{document}